\documentclass{llncs}
\usepackage[utf8]{inputenc}
\usepackage{enumerate}
\usepackage{graphicx}
\usepackage{epstopdf}
\usepackage{amsfonts}
\usepackage{amssymb}
\usepackage{amsmath}
\usepackage[labelfont=bf]{caption}
\usepackage{sectsty}
\usepackage{wasysym}
\usepackage{hyperref}
\usepackage{parskip}

\let\doendproof\endproof
\renewcommand\endproof{~\hfill\qed\doendproof}

\def\computationproblem#1#2#3{
  \begin{center}
  \begin{tabular}{rp{0.8\textwidth}}
  {\sc Problem:\enspace}&#1\\
  {\sc Input:\enspace}&#2\\
  {\sc Question:\enspace}&#3\\
  \end{tabular}
  \end{center}
}

\def\cNP{\hbox{\rm \sffamily NP}}

\def\int{\hbox{\bf \rm \sffamily INT}}

\def\ca{\hbox{\bf \rm \sffamily CA}}

\begin{document}

\title{A note on simultaneous representation problem for interval and circular-arc graphs}
\author{Jan Bok \inst{1}
\and Nikola Jedličková \inst{2}}
\institute{
Computer Science Institute of Charles University, Malostransk\'{e} n\'{a}m\v{e}st\'{i} 25, 11800, Prague, Czech Republic. Email: \email{bok@iuuk.mff.cuni.cz}
\and
Department of Applied Mathematics, Charles University, Malostransk\'{e} n\'{a}m\v{e}st\'{i} 25, 11800, Prague, Czech Republic. Email: \email{jedlickova@kam.mff.cuni.cz}
}

\maketitle

\begin{abstract}
In this short note, we show two \cNP-completeness results regarding the \emph{simultaneous representation problem}, introduced by Lubiw and Jampani \cite{jampani2009simultaneous,jampani_interval}.

The simultaneous representation problem for a given class of intersection graphs asks if
some $k$ graphs can be represented so that every vertex is represented by the
same interval in each representation. We prove that it is \cNP-complete to decide this for 
the class of interval and circular-arc graphs in the case when $k$ is a part of the input
and graphs are not in a sunflower position.
\end{abstract}

\section{Introduction}

Jampani and Lubiw introduced \emph{the simultaneous representation problem}
defined for any class of intersection graphs in
\cite{jampani2009simultaneous}.

\begin{definition}
Let $\mathcal{C}$ be a class of intersection graphs. Graphs 
$G_1, \ldots, G_k \in \mathcal{C}$ are simultaneously representable (or simultaneous) if there exist representations $R_1, \ldots, R_k$ of $G_1, \ldots, G_k$ such that 
$$\forall i,j \in \{1, \ldots, k\} \ \forall v \in G_i \cap G_j: R_i(v) = R_j(v).$$
\end{definition} 

The \emph{simultaneous representation problem} for a class $\mathcal{C}$ asks
if given $k$ graphs $G_1, \ldots, G_k \in \mathcal{C}$ are simultaneous. We
distinguish whether $k$ is fixed or if it is a part of the input. The problem can be further
divided into two cases, depending on whether graphs are in 
\emph{sunflower position} or not.

\begin{definition}
We say that graphs $G_1, \ldots, G_k$ are in a \emph{$k$-sunflower position}
if there exists a set of vertices $I$ such that $G_i \cap G_j = I $ for every
$i \neq j$. Otherwise we say that these graphs are in a non-sunflower
position.
\end{definition}

Jampani and Lubiw
were first to consider this problem and solved it for chordal, comparability, permutation
\cite{jampani2009simultaneous} (in sunflower configuration for both possibilities -- $k$ fixed and $k$ on the input) and then for interval graphs \cite{jampani_interval} (sunflower position and $k = 2$). The algorithm for 2 interval
graphs was later improved by Bl{\"a}sius and Rutter.

A summary of known results for various type of graphs in sunflower position is in
Table \ref{t}. We point the reader to the PhD thesis of Jampani \cite{jampani2011simultaneous} for a broader introduction to simultaneous
representation problem.

\vspace{-0.2cm}
\begin{table}
\begin{tabular}{|l|l|l|}
\hline
\textbf{class of graphs} & \textbf{2 graphs} & \textbf{$k$ graphs, $k$ not fixed} \\ \hline
chordal graphs & $O(n^3)$ \cite{jampani2009simultaneous}         & \cNP-hard \cite{jampani2009simultaneous}          \\ \hline
interval graphs          & $O(n^2 \log n)$ \cite{jampani_interval}, improved to $O(n)$ in \cite{blasius2016simultaneous}          &  open         \\ \hline
comparability graphs          & $O(nm)$ \cite{jampani2009simultaneous}         & $O(nm)$ \cite{jampani2009simultaneous}         \\ \hline
permutation graphs          & $O(n^3)$ \cite{jampani2009simultaneous}        & $O(n^3)$ \cite{jampani2009simultaneous}         \\ \hline
\end{tabular}
\vspace{0.2cm} 
\caption{A summary of known results for graphs in sunflower position.}
\label{t}
\end{table}
\vspace{-1.2cm}

\section{Results}

We show that the simultaneous representation problem for $k$ interval or
$k$ circular-arc graphs in non-sunflower position where $k$ is a part of the
input is \cNP-complete.

For both \cNP-completeness results we use a reduction from
\textsc{TotalOrdering}. Opatrny proved that \textsc{TotalOrdering} is
\cNP-complete \cite{opatrny1979total}.

\computationproblem
{\textsc{TotalOrdering} -- Total Ordering Problem}
{A finite set S and a finite set T of triples from S}
{Does there exist a total ordering $<$ of S such that for all triples $(x,y,z) \in T$ either $x < y < z$ or $x > y > z$?}

\subsection{Interval graphs}

\computationproblem
{\textsc{SimRep}(\int) -- Simultaneous representation problem for interval graphs}
{Interval graphs $G_1, \ldots, G_k$.}
{Do there exist interval representations $R_1, \ldots, R_k$ of $G_1, \ldots, G_k$ such that $$\forall i,j \in \{1, \ldots, k\} \ \forall v \in G_i \cap G_j: R_i(v) = R_j(v)?$$}

\begin{theorem}
\textsc{SimRep}(\int) for $k$ interval graphs in a non-sunflower position where $k$ is not fixed is \cNP-complete.
\end{theorem}

\begin{proof}
The problem is clearly in \cNP\ as we can easily check in polynomial time if given
representations are simultaneous.

Now, let $I_{\textrm{TO}}$ be an instance of \textsc{TotalOrdering}. Let us
set $s:= |S|$ and $t := |T|$. We denote by $(x_i, y_i, z_i)$ each triple for
$i \in \{1, \ldots, t\}$.  We will construct an instance $I_S$ of
\textsc{SimRep}(\int).

We define graphs $G_0, G_1, \ldots, G_t$ in the following way.
\begin{itemize}
\setlength{\itemsep}{-5pt}
\item $G_0 := (S , \emptyset)$
\item $G_i := (V_i, E_i)$ for each $0 < i \leq t,$ where 
\begin{itemize}
\setlength{\itemsep}{-4pt}
\item $V_i := \{x_i, y_i, z_i, a_i, b_i, c_i\},$
\item $E_i := \{x_ib_i, y_ib_i, z_ib_i, x_ia_i, z_ic_i\}$ 
\end{itemize}
\end{itemize}

We observe that graphs $G_i$ are interval graphs and thus this is indeed an
instance of \textsc{SimRep}(\int). See Example \ref{example} and Figure
\ref{fig:example1} for an illustration of this construction.

Now we can check that the following holds.
\begin{itemize}
\setlength{\itemsep}{-4pt}
\item
$G_0 \cap G_i = \{x_i, y_i, z_i\}$ for each $1 \leq i \leq t$,
\item
$G_i \cap G_j = G_i \cap G_j \cap G_0 = \{v \in S | v \in (x_i, y_i, z_i) \land v \in (x_j, y_j, z_j)\}$ for each $1 \leq i < j \leq t$.
\end{itemize}

We observe that these $t+1$ graphs are simultaneous if and only if the original instance of
\textsc{TotalOrdering} has a solution. We can read the linear order $<$ of $S$ from the representation of its corresponding vertices in $G_0$.

Thus \cNP-completeness is established.
\end{proof}

\begin{example} \label{example}
For an instance $I_{TO}$ of \textsc{TotalOrdering} where
\begin{align*}
S &= \{1,2,3,4,5\}, \\
T &= \{(5,1,2), (2,4,3), (1,4,3)\},
\end{align*}
we build an instance $I_S$ of \textsc{SimRep}(\int) as in Figure \ref{fig:example1}.
\end{example}

\begin{figure}
\centering
\includegraphics[scale=1.1]{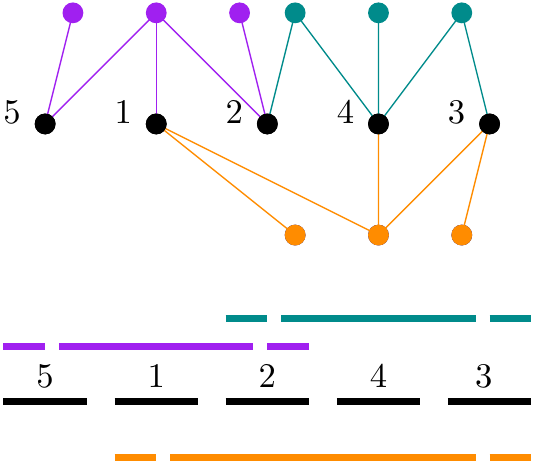}
\caption{In the top: A union of graphs for triples from Example \ref{example}. In the bottom: Their simultaneous interval representations.}
\label{fig:example1}
\end{figure}

\subsection{Circular-arc graphs}

\computationproblem
{\textsc{SimRep}(\ca) -- Simultaneous representation problem for circular-arc graphs}
{Circular-arc graphs $G_1, \ldots, G_k$.}
{Do there exist circular-arc representations $R_1,\ldots, R_k$ of $G_1,\ldots,G_k$ such that
$$\forall i,j \in \{1, \ldots, k\} \ \forall v \in G_i \cap G_j: R_i(v) = R_j(v)?$$}
\vspace{-15pt}
\begin{theorem}
\textsc{SimRep}(\ca) for $k$ graphs in a non-sunflower
position where $k$ is not fixed is \cNP-complete.
\end{theorem}

\begin{proof}
Again, problem is in \cNP\ from obvious reasons.

We will proceed in a similar way as in the previous proof. We define graphs
$G_0, G_1, \ldots, G_t$ in the same way as before and we add one common
isolated vertex $x$ to every graph, i.e. $x \in \bigcap_{i=0}^t G_i$. Vertex
$x$ takes the role of breaking the cycle into a segment and thus we can
argue the rest as for interval graphs.
\end{proof}

\bibliographystyle{plain}
\bibliography{simrep_note}

\end{document}